\documentclass[pmlr]{jmlr}

\usepackage{longtable}
\usepackage{tikz}
\usetikzlibrary{positioning}
\usetikzlibrary{arrows,automata}
\usepackage{verbatim}
\usepackage{mathtools}
\usepackage{color}
\usepackage[T1]{fontenc}
\usepackage[utf8]{inputenc}
\usepackage{amssymb}
\usepackage{floatrow}
\usepackage[normalem]{ulem}
\usepackage{soul}
\usepackage{bbm}

\newfloatcommand{capbtabbox}{table}[][\FBwidth]
\usepackage{graphics}
\usepackage{blindtext}

\usepackage{booktabs}
\usepackage[load-configurations=version-1]{siunitx} 

\newcommand{\wver}{q_{a\downarrow}}
\newcommand{\whor}{q_{a\rightarrow}}
\newcommand{\bhor}{q_{b\rightarrow}}
\newcommand{\bver}{q_{b\downarrow}}

\renewcommand{\enspace}{}
\newcommand{\R}{\Rbb}

\newcommand{\Rbb}{\mathbb{R}}

\newcommand{\Fbb}{\mathbb{F}}

\newcommand{\Fcal}{{\mathcal{F}}}

\renewcommand{\vec}[1]{\ensuremath{\mathbf{#1}}}
\newcommand{\vecs}[1]{\ensuremath{\mathbf{\boldsymbol{#1}}}}
\newcommand{\mat}[1]{\ensuremath{\mathbf{#1}}}

\newcommand{\ten}[1]{\mat{\ensuremath{\boldsymbol{\mathcal{#1}}}}}

\newcommand{\A}{\mat{A}}
\newcommand{\B}{\mat{B}}

\newcommand{\M}{\mat{M}}

\renewcommand{\P}{\mat{P}}

\newcommand{\x}{\vec{x}}

\DeclareMathOperator*{\Tr}{Tr} 

\newcommand{\nstates}{n}

\newcommand{\szerosymbol}{\alpha}
\newcommand{\szero}{\vecs{\szerosymbol}}

\newcommand{\sinfsymbol}{\omega}
\newcommand{\sinf}{\vecs{\sinfsymbol}}

\DeclareDocumentCommand{\wa}{  O{A} O{\Fbb^\nstates} O{\szero} O{\sinf} }%
{(#2,#3,\{\mat{#1}^\sigma\}_{\sigma\in\Sigma},#4)}
\DeclareDocumentCommand{\waR}{  O{A} O{\Rbb^\nstates} O{\szero} O{\sinf} }%
{(#2,#3,\{\mat{#1}^\sigma\}_{\sigma\in\Sigma},#4)}

\newcommand{\tzerosymbol}{\alpha}
\newcommand{\tzero}{\vecs{\tzerosymbol}}

\newcommand{\tinfsymbol}{\omega}
\newcommand{\tinf}{\vecs{\tinfsymbol}}

\DeclareDocumentCommand{\wta}{ O{T} O{\Rbb^\nstates} O{\tzero} O{\tinf} O{\Fcal}}%
{(#2,#3,\{\ten{#1}^g\}_{g\in #5_{\geq 1}},\{#4^\sigma\}_{\sigma\in #5_0})}

\DeclareDocumentCommand{\trees}{g}{\IfNoValueTF{#1}{\mathfrak{T}}{\mathfrak{T}_{#1}}}
\DeclareDocumentCommand{\contexts}{g}{\IfNoValueTF{#1}{\mathfrak{C}}{\mathfrak{C}_{#1}}}

\newcommand{\gwmprod}{\diamond}

\DeclareDocumentCommand{\gwm}{  O{M} O{\Fbb^\nstates}}{(#2, \{\ten{#1}^x\}_{x\in\Sigma})}
\DeclareDocumentCommand{\gwmcirc}{  O{M} O{\Rbb^\nstates}}{(#2, \{\mat{#1}^\sigma\}_{\sigma\in\Sigma})}
\DeclareDocumentCommand{\dgwm}{ O{M} O{\Fbb^\nstates}}{(#2, \{\ten{#1}^x\}_{x\in\Sigma},\gwmprod)}

\theorembodyfont{\upshape}
\theoremheaderfont{\scshape}
\theorempostheader{:}
\theoremsep{\newline}

\jmlrvolume{93}
\jmlryear{2018}
\jmlrworkshop{International Conference on Grammatical Inference 2018}

\title[Learning GWMs on pictures]{Learning Graph Weighted Models on Pictures}

  \author{\Name{Philip Amortila} \Email{philip.amortila@mail.mcgill.ca} \\
  \Name{Guillaume Rabusseau} \Email{guillaume.rabusseau@mail.mcgill.ca}\\
   \addr Reasoning and Learning Lab - School of Computer Science - McGill University}

\begin{document}

\newpage

\maketitle

\begin{abstract}

Graph Weighted Models (GWMs) have recently been proposed
as a natural generalization of weighted automata over strings and trees
 to arbitrary families of labeled graphs (and hypergraphs).
A GWM generically associates a labeled graph with a tensor
network and computes a value by successive contractions directed by
its edges. In this paper, we consider the problem of learning GWMs defined over
the graph family of pictures~(or 2-dimensional words). As a proof of concept, we
consider regression and classification tasks over the simple \emph{Bars \& Stripes} and \emph{Shifting Bits} picture languages and provide an
experimental study investigating whether these languages can be learned in the form of a GWM from 
positive and negative examples using gradient-based methods. Our results suggest that this
is indeed possible and that investigating the use of gradient-based methods to learn picture series
and functions computed by GWMs over other families of graphs could be a fruitful direction.

\end{abstract}

\section{Introduction}
\label{sec:intro}

The heart of automata theory is the modeling and study of functions defined over syntactical structures such as strings, trees, or graphs. A classic example are Weighted Finite Automata~(WFA)~\citep{schutzenberger1961definition,berstel1988rational}, which are finite state machines computing real-valued functions over strings. WFA 
encompass a wide class of useful tools for predictions such as hidden Markov models, predictive state representations, and probabilistic automata and are thus particularly relevant to the machine learning community. 
Of great interest, specifically for machine learning, are the so-called spectral learning algorithms~\citep{bailly2009grammatical,hsu2009spectral,balle2014spectral}, which are efficient and consistent algorithms for learning WFAs. 

Various extensions of WFA have been proposed such as Weighted Tree Automata~(WTA) \citep{Berstel82} and Weighted Picture Automata (WPA)~\citep{Bozapalidis_Grammatikopoulou_2005}, which model functions defined over trees and pictures respectively. 
Graph Weighted Models (GWM), introduced in \citep{bailly2015recognizable,JCSS2017}, are a natural generalization to the  case of labelled graph inputs, and include the above models as special cases. 
Roughly speaking, a GWM is defined by associating a tensor with every character of a ranked alphabet and computes a real-valued function over graphs labelled by symbols of the alphabet.
The value computed by a GWM is obtained by constructing a \textit{tensor network} out of the graph input and computing a value via successive tensor contractions.

To date, GWMs have been studied mostly from the perspective of formal languages~\citep{JCSS2017,FOSSACS2018}, and designing learning algorithms for these automata remains to be done. In particular, extending the spectral algorithms for WFA to GWMs remains an open (and challenging) problem which we are currently investigating. 
In this work, we examine an alternative approach for learning GWMs from data by experimentally exploring how traditional gradient-based algorithms perform on this task. While this approach can be applied to GWMs defined over arbitrary families of graphs, we focus here on regression and classification problems for simple picture languages. Firstly, as an instance of regression tasks, we examine the \emph{Bars \& Stripes} language made of pictures containing only horizontal or vertical stripes. After showing that this language can be computed by a WPA~(and equivalently by a GWM defined over pictures), we show empirically that minimizing the mean squared error over a large enough sample of input/output examples allows one to recover a GWM approximating the function of interest which generalizes to unseen pictures of different sizes. Secondly, akin to logistic regression, we show how to handle classification tasks by composing the output of the model with a sigmoid activation function, and empirically demonstrate that minimizing the cross-entropy error can lead to successful classification of the \emph{Shifting Bits} language, which consists of pictures where each row is a horizontal shift of the previous one.

\section{Preliminaries}
\label{sec:prelim}

Firstly, we introduce the necessary notions. In this paper we are working with picture data, thus we begin by discussing Weighted Picture Automata.

\subsection{Weighted Picture Automaton}

A \emph{picture}~(also called an \textit{image} or a \textit{2d-word})  $p\in\mathcal{P}$ over a finite alphabet $\Sigma$ is defined as a non-empty rectangular 
array of elements of $\Sigma$, formally $\mathcal{P} = \cup_{m,n\geq 1} \Sigma^{m\times n}$. 
Given $p \in \Sigma^{m \times n}$, we write $p_{i,j}$ for the component of $p$ at position $(i,j)$. A \emph{picture language} is a set of pictures, 
while a \emph{picture series} is a function from $\mathcal{P}$ to a commutative semi-ring. Regular 
picture languages can equivalently be described in terms of automata, sets of tiles, rational operations, or monadic 
second order logic~\citep{Giammarresi_Restivo_1996,Giammarresi_Restivo_Seibert_Thomas_1996,INOUE_NAKAMURA_1977,Latteux_Simplot_1997}.
The extension of regular picture languages to the quantitative setting led to the definition of \emph{recognizable picture series} 
whose theoretical study has been of recent interest~%
\citep{Bozapalidis_Grammatikopoulou_2005,Maurer_2005,fichtner2011weighted,Babari_Droste_2015}. 
Recognizable picture series were first introduced in \citep{Bozapalidis_Grammatikopoulou_2005} 
by means of Weighted Picture Automata.

\begin{definition}[\cite{Bozapalidis_Grammatikopoulou_2005}]\label{def:wpa}
A Weighted Picture Automaton (WPA) over an alphabet $\Sigma$ is a tuple $A=\left(Q,R,F_w,F_n,F_e,F_s,\delta\right)$, where $Q$ is a finite set of states, $R \subseteq \Sigma \times Q^4$ is a finite set of rules, $F_w,F_n,F_e,F_s \subseteq Q$ are four poles of acceptance, and $\delta : R \rightarrow \mathbb{R}$ is the weighted transition function%
\footnote{WPAs are originally defined over arbitrary commutative semi-rings but we will only consider the field $\mathbb{R}$.}.
Given a rule $r=\left(\sigma,q_w,q_n,q_e,q_s\right) \in R$, we denote its label by $\ell(r) = \sigma$, and its poles by $west(r) = q_w,\  north(r)=q_n, \ east(r)=q_e,$  and $south(r)=q_s$. Given an image $p \in \Sigma^{m \times n}$, a run $c \in R^{m \times n}$ is an assignment of rules such that the following compatibility properties hold:
\begin{align*}
 south(c_{i,j}) &=north(c_{i+1,j}) \ \forall i \in [m-1], \forall j \in [n] \\  
 east(c_{i,j}) &= west(c_{i,j+1}) \ \forall i \in [m], \forall j \in [n-1]
\end{align*}
and $\ell(c_{i,j}) = p_{i,j} \ \forall i \in [m], j \in [n]$, where $[n] \coloneqq \{1,...,n\}$. A run is \textit{accepted} if the outer poles are in their respective poles of acceptance, i.e. 
\begin{equation*}
west(c_{i,1}) \in F_w, \ south(c_{m,j}) \in F_s,\  east(c_{i,n}) \in F_e, \ north(c_{1,j}) \in F_n \ \forall i \in [m], \forall j \in [n] \ 
\end{equation*}
We denote by $R(p)$ the set of accepted runs on $p$. The weight function is extended to runs via $\delta(c) = \Pi_{i,j} \delta(c_{i,j})$. The function computed by A is the sum of the weights over all accepting runs: $f_A(p) = \sum_{c \in R(p)} \delta(c)$, with the convention that $f_A(p) = 0$ if there are no accepting runs.
\end{definition}

\subsection{Tensor Networks and Graph Weighted Models}
We first briefly introduce \emph{tensors}
and \emph{tensor networks}. For our purposes, a tensor over $\mathbb{R}$ of order $p$ can be seen as a $p$-dimensional array of scalars: $\left(\mathcal{T}_{i_1,i_2,..,,i_p} \in \mathbb{R} : i_n \in [d], n \in [p]\right)$\footnote{In general we could have $i_n \in [d_n]$ for any list of integers $d_1,d_2,...,d_p$, however here we only consider the case $d_n=d \ \forall n$ -- these are called \textit{hypercubic} tensors.}. Tensor networks are a useful tool for visualizing operations on tensors and will simplify our exposition of GWMs. 
A tensor network is an undirected graph where the nodes correspond to tensors and the outgoing edges correspond to the different dimensions of the tensor. For example, a vector is a node of degree 1 and a matrix is node of degree 2. We allow nodes to have free edges that do not connect to other nodes. Each edge is numbered by an index corresponding to a dimension of the tensor, such a numbered edge is called a \emph{port}. Connecting two ports in a tensor network corresponds to a summation of the two tensors along the connected indices. A few examples of tensor networks and their associated computations are shown in Figure~\ref{fig:tn}. 

\begin{figure}[h!]
	\centering
	\vspace{\baselineskip}
	\hspace*{-1.5cm}\begin{tikzpicture}
		\input{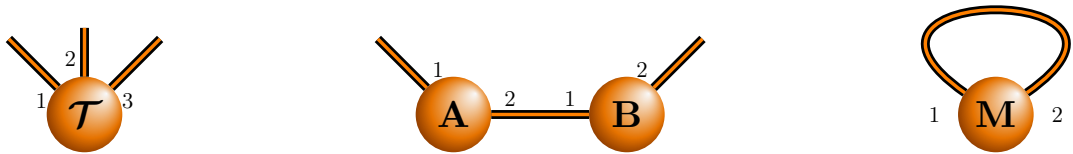}
		\scalebox{1.3}{
		\node[tensor](T){$\ten{T}$};
		\node[draw=none, above left= 0.5cm and 0.5cm of T](T-1) {};
		\node[draw=none, above=0.5cm of T](T-2) {};
		\node[draw=none, above right=0.5cm and 0.5cm of T](T-3) {};
		\simpleedgeports{T}{1}{below left}{T-1}{}{}
		\simpleedgeports{T}{2}{above left}{T-2}{}{}
		\simpleedgeports{T}{3}{below right}{T-3}{}{}

		\node[tensor,right=3cm of T](A){$\mat{A}$};
		\node[tensor, right=of A](B){$\mat{B}$};
		\node[draw=none, above left=0.5cm and 0.5cm of A](A-1) {};
		\node[draw=none, above right=0.5cm and 0.5cm of B](B-2) {};
		\node[tensor,right=3cm of B](M){$\mat{M}$};
		\draw[edge] (M) to [in=145,out=35,looseness=8] (M) node[port,left=0.1cm of M]{1} node[port, right=0.1cm of M]{2};
		\simpleedgeports{A-1}{}{}{A}{1}{above right}
		\simpleedgeports{A}{2}{above right}{B}{1}{above left}
		\simpleedgeports{B-2}{}{}{B}{2}{above left}
		}
	\end{tikzpicture}
	\caption{(left) The arity of a node is the order of the tensor, here the node denotes a 3rd-order tensor $\ten{T} \in \R^{d\times d\times d}$.
	(middle) The connection of two ports corresponds to a contraction over the corresponding indices, here connecting the second port of the matrix $\A$ with the first port of the
	matrix $\B$ represents  the classical matrix product: $(\mat{AB})_{i_1,i_2} = \sum_k \mat{A}_{i_1,k}\mat{B}_{k,i_2}$. (right) Similarly, connecting the two ports of a
	matrix represents the trace operation: $\Tr(\M)  = \sum_k \M_{kk}$.}
	\label{fig:tn}
\end{figure}
\noindent Formally, one way to compute the tensor represented by a tensor network by
first taking the tensor product of the tensors associated with all the nodes in the graph,
and then performing the contractions associated to the edges, i.e. summing over the pairs of indices corresponding to connected ports. For example, for the tensor network
in Figure~\ref{fig:tn}~(middle), the tensor product of $\A$ and $\B$ is given by the fourth order tensor $(\A\otimes \B)_{i_1i_2i_3i_4}=\A_{i_1i_2}\B_{i_3i_4}$, and one then performs the contraction
between the second port of $\A$~($i_2$) and the first port of $\B$~($i_3$) to obtain the second order tensor $\sum_{j}\A_{i_1j}\B_{ji_4}$.
Note first that in a graph with more than one edge the order of contractions does not matter; second, this is a naive way of performing the computation:  one does not need 
to construct the whole tensor product before performing the contractions, however finding the optimal contraction sequence is an NP-hard problem~\citep{pfeifer2014faster}.

Graph Weighted Models have been introduced as computational models over labelled graphs in \citep{bailly2015recognizable,JCSS2017}. A Graph Weighted Model associates a labelled graph with a tensor network computing the corresponding value. Here we present only the specifics necessary for computations of GWMs over pictures. 

\begin{definition}[Graph Weighted Models]\label{def:gwm} 
A GWM $M$ (over $\mathbb{R}$) \textit{on pictures} is a tuple $M=\left(d, \{\mathcal{T}^\sigma\}_{\sigma \in \Sigma}, \alpha^w,\alpha^n,\alpha^e,\alpha^s\right)$, where $d$ is the dimension (or number of states), $\mathcal{T}^\sigma \in \mathbb{R}^{d\times d\times d\times d}$ is a tensor of order 4  for each $\sigma \in \Sigma$, and $\alpha^w,\alpha^n,\alpha^e,\alpha^s \in \mathbb{R}^d$ are the border vectors. Given a picture $p$, the function computed by $M$ is the tensor network obtained from $p$ by replacing every character with its associated tensor and adding the border vectors.
\end{definition}

As an example, the value computed by a GWM $M$ on the picture $p={a\ b \atop b\ a}$ is represented as a tensor network in Figure~\ref{fig:TN.pic}.
\begin{figure}[h!]
	\centering
	\begin{tikzpicture}
	\input{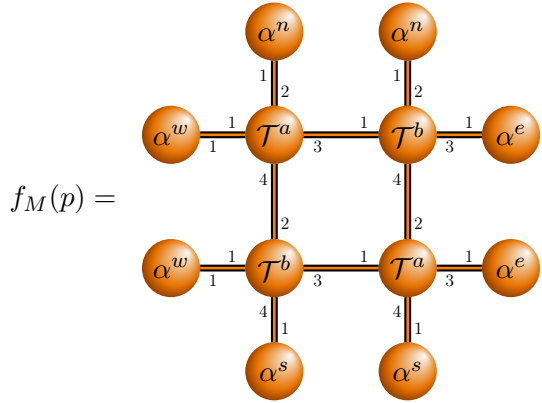}
\scalebox{1}{  

  \node[tensor] (v11)  {$\mathcal{T}^a$};
  \node[tensor, right = 1cm of v11] (v12)  {$\mathcal{T}^b$};
  \node[tensor, below = 1cm of v11] (v21)  {$\mathcal{T}^b$};
  \node[tensor, right = 1cm of v21] (v22)  {$\mathcal{T}^a$};
  
 \node[tensor, left = 0.6cm of v11] (v10) {$\alpha^w$}; 
 \node[tensor, left = 0.6cm of v21] (v20) {$\alpha^w$}; 
  \node[draw=none, above left =0.45cm of v20] (txt) {$f_M(p)=$};
 \node[tensor, above = 0.6cm of v11] (v01) {$\alpha^n$}; 
 \node[tensor, above = 0.6cm of v12] (v02) {$\alpha^n$}; 
 \node[tensor, right = 0.6cm of v12] (v14) {$\alpha^e$}; 
 \node[tensor, right = 0.6cm of v22] (v24) {$\alpha^e$}; 
 \node[tensor, below = 0.6cm of v21] (v31) {$\alpha^s$}; 
 \node[tensor, below = 0.6cm of v22] (v32) {$\alpha^s$}; 
 
 \simpleedgeports{v10}{1}{below right}{v11}{1}{above left}
 \simpleedgeports{v20}{1}{below right}{v21}{1}{above left}
 \simpleedgeports{v01}{1}{below left}{v11}{2}{above right}
 \simpleedgeports{v02}{1}{below left}{v12}{2}{above right}
 \simpleedgeports{v12}{3}{below right}{v14}{1}{above left}
 \simpleedgeports{v22}{3}{below right}{v24}{1}{above left}
 \simpleedgeports{v21}{4}{below left}{v31}{1}{above right}
 \simpleedgeports{v22}{4}{below left}{v32}{1}{above right}
  
 \simpleedgeports{v11}{3}{below right}{v12}{1}{above left}
 \simpleedgeports{v21}{3}{below right}{v22}{1}{above left}

 \simpleedgeports{v11}{4}{below left}{v21}{2}{above right}
 \simpleedgeports{v12}{4}{below left}{v22}{2}{above right}
  }
	\end{tikzpicture}	
\caption{The tensor network associated with the picture $p={a\ b \atop b\ a}$. Note that the tensor network has no free edges, thus indeed computes to a real number.}\label{fig:TN.pic}
\end{figure}

\noindent The computation can be written explicitly as:
\begin{equation}\label{eq:gwm}
f_M(p) = \sum_{i_1,i_2,..,i_{12}}\alpha^w_{i_1}\alpha^n_{i_2}\mathcal{T}^a_{i_1,i_2,i_3,i_4}\mathcal{T}^b_{i_3,i_5,i_6,i_7}\alpha^n_{i_5}\alpha^e_{i_6}\mathcal{T}^b_{i_8,i_4,i_9,i_{10}}\alpha^w_{i_8}\alpha^s_{i_{10}}\mathcal{T}^a_{i_9,i_7,i_{11},i_{12}}\alpha^e_{i_{11}}\alpha^s_{i_{12}}
\end{equation}

\noindent The following useful result states that any WPA can be realized by a GWM on pictures.

\begin{proposition}{\cite[Proposition 4]{rabusseau2016thesis}}\label{prop:wpa}
Any function that is computable by a WPA with $d$ states is computable by a $d$-dimensional GWM on pictures.
\end{proposition}
We note that the converse of the previous proposition holds in the sense that WPA can compute
any function that can be computed by a GWM defined on the family of graph representations of 
pictures~(see~\citep[Proposition~3.9]{JCSS2017} for more details).

\section{Learning of GWMs on Pictures}
\label{sec:learning}

In this section we present the languages we will attempt to learn in our experiments. Afterwards, we formalize the learning problem, and present the gradient--based approach we propose to tackle it. 

\subsection{Bars \& Stripes and Shifting Bits}\label{sec:bars}

We will apply our methods to regression and classification tasks on two picture languages over 2-letter alphabets (denoted $a$ for white and $b$ for black), demonstrating the ability of our approach to learn simple picture languages.

\subsubsection{Bars \& Stripes}

In the \textit{Bars \& Stripes} (BS)  language~\citep{mackay2003information}, each image is composed of either horizontal stripes or vertical bars, but not both~(unless it is fully white or fully black). Formally, we have $BS = \{p \in \mathcal{P}\colon (p_{i,j} = p_{i+1,j}\ \forall  i,j ) \text{ or } (p_{i,j} = p_{i,j+1}\ \forall  i,j )\} \subseteq \mathcal{P}$. A few sample images are shown in Figure~\ref{fig:BS}.
\begin{figure}[h]
\begin{center}
\includegraphics[width=0.8\textwidth]{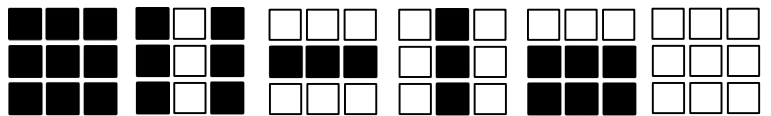}
\end{center}
\caption{Examples of $3\times 3$ images in the Bars \& Stripes language. This figure and the next from~\citep{melchior2016center}}\label{fig:BS}
\end{figure}
\subsubsection{Shifting Bits}
In the \textit{Shifting Bits} (SB) language, inspired from~\citep{melchior2016center}, each row in the pictures is a horizontal translation of the previous row. We allow for shifts of arbitrary length. Formally, we have $SB = \cup_{m,n} SB_{m \times n}$, where $SB_{m \times n} = \{p \in \Sigma^{m \times n}: \exists s \in [m] \text{ such that } (p_{i+1,j} = p_{i,j-s} \text{ if } j-s \geq 1 \text{ else } p_{i+1,j} = b)\}$. Some examples of shift images are given in figure \ref{fig:SB}.

\begin{figure}[h]
\begin{center}
\includegraphics[width=0.8\textwidth]{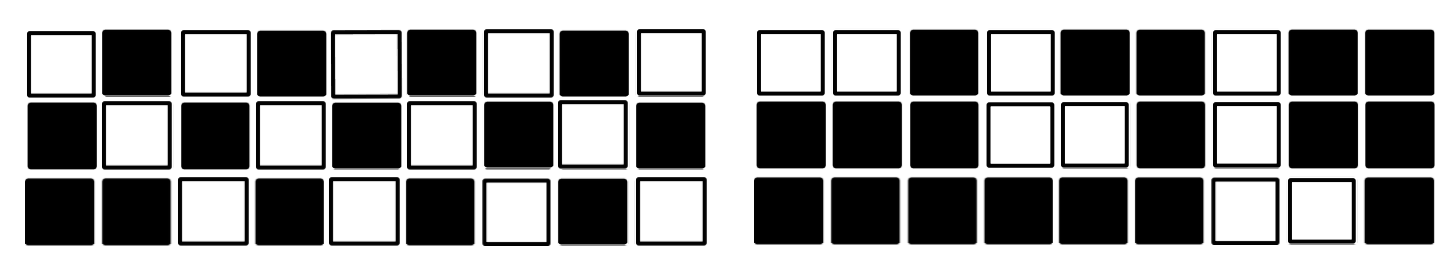}
\end{center}

\caption{Examples of $3 \times 9$ images in the Shifting Bits language, with shifts of size $1$ and $3$.}\label{fig:SB}
\end{figure}
\subsection{Recognizability}
Before considering the problem of learning these languages from positive and negative examples, 
we first show that the Bars \& Stripes language is indeed recognizable by a WPA. 
\begin{proposition}[BS is WPA-recognizable]\label{prop:bs}
The Bars \& Stripes language $BS \subseteq \mathcal{P}$ is WPA-recognizable, in the sense that there exists a WPA whose support is the Bars \& Stripes language. 

More precisely, there exists a WPA $A$ with $6$ states such that
$$f_A(p) = \begin{cases}
1&\text{if } p\in \mathcal{P}\text{ and }p\text{ contains at least one } a \text{ and one }b\\
2&\text{if } p\in \mathcal{P}\text{ and }p \text{ contains only }a\text{'s or only }b\text{'s}\\
0&\text{otherwise}.
\end{cases}$$
\end{proposition}
\begin{proof}
Consider the WPA $A=(Q,R, F_w,F_e,F_s,F_n,\delta)$, where  $Q = \{q_0,q_{a\rightarrow},q_{a\downarrow},q_{b\rightarrow},q_{b\downarrow},q_f\}$, $F_w=F_n =q_0$, and $F_e = F_s = q_f$. Intuitively, $q_0$ and $q_f$ are the initial and final states, the horizontal states $q_{a\rightarrow}, q_{b\rightarrow}$ encode whether we are in a picture with horizontal bars, and similarly for the vertical states $q_{a\downarrow},q_{b\downarrow}$. For succinctness, we represent the rule set $R$ via the following diagrams:

\begin{tikzpicture}[scale=.9, transform shape]
\tikzstyle{a} = [circle, fill=gray]
\node (a1) at (-1.5, 1.5) {a};
\coordinate (b1) at +(-1.5,3);
\coordinate (c1) at +(0, 1.5);
\coordinate (d1) at (-3,1.5);
\coordinate (e1) at (-1.5,0);
\draw[thick,->] (b1) -- (a1) node[midway, right] {$q_0, \whor,\bhor  $};
\draw[thick,->] (d1) -- (a1) node[midway,above] {$q_0,\whor $};
\draw[thick,->] (a1) -- (c1) node[midway,above] {$q_f,\whor$};
\draw[thick,->] (a1) -- (e1) node[midway,right] {$q_f,\whor$};

\node (a2) at (3, 1.5) {a};
\coordinate (b2) at +(3,3);
\coordinate (c2) at +(4.5, 1.5);
\coordinate (d2) at (1.5,1.5);
\coordinate (e2) at (3,0);
\draw[thick,->] (b2) -- (a2) node[midway, right] {$q_0, \wver $};
\draw[thick,->] (d2) -- (a2) node[midway,above] {$q_0,\wver,\bver $};
\draw[thick,->] (a2) -- (c2) node[midway,above] {$q_f,\wver$};
\draw[thick,->] (a2) -- (e2) node[midway,right] {$q_f,\wver$};

\node (a3) at (7.5, 1.5) {b};
\coordinate (b3) at +(7.5,3);
\coordinate (c3) at +(9, 1.5);
\coordinate (d3) at (6,1.5);
\coordinate (e3) at (7.5,0);
\draw[thick,->] (b3) -- (a3) node[midway, right] {$q_0, \whor, \bhor $};
\draw[thick,->] (d3) -- (a3) node[midway,above] {$q_0,\bhor $};
\draw[thick,->] (a3) -- (c3) node[midway,above] {$q_f,\bhor$};
\draw[thick,->] (a3) -- (e3) node[midway,right] {$q_f,\bhor$};

\node (a4) at (12, 1.5) {b};
\coordinate (b4) at +(12,3);
\coordinate (c4) at +(13.5, 1.5);
\coordinate (d4) at (10.5,1.5);
\coordinate (e4) at (12,0);
\draw[thick,->] (b4) -- (a4) node[midway, right] {$q_0, \bver $};
\draw[thick,->] (d4) -- (a4) node[midway,above] {$q_0,\bver, \wver $};
\draw[thick,->] (a4) -- (c4) node[midway,above] {$q_f,\bver$};
\draw[thick,->] (a4) -- (e4) node[midway,right] {$q_f,\bver$};

\end{tikzpicture}

The diagrams indicate that when the label in the center is observed, any combination of the above states are valid for the poles of the rule. The first picture, for example, says that when the white label is observed, a valid run has poles $west(a) \in \{q_0,\whor\},\  north(a) \in \{q_0,\whor, \bhor\},\  east(a) \in \{q_f, \whor\},\ \text{and } south(a) \in \{q_f,\whor\}$. Finally, the transition function is simply $\delta(r) = 1 \ \forall r \in R$. Thus, $\delta(c) = \Pi_{i,j} \delta(c_{i,j}) = 1$ for an accepting run and $f_A(p)$ simply counts the number of accepting runs on an image $p$.

First we show that if $p$ is an image in BS, then there exists an accepting run for $p$. Assume $p \in \{a,b\}^{m \times n}$ is horizontally striped, and let $\sigma_i$ be the common label of the $i^\text{th}$ row. Consider the run $c$ defined by
\begin{align*}
north(c_{1,j}) &= west(c_{i,1}) = q_0 \ \forall i \in [m] ,j \in [n] \\
east(c_{i,n}) &= south(c_{m,j}) = q_f \ \forall i \in [m],j \in [n] \\
east(c_{i,j}) &= south(c_{i,j}) = q_{\sigma_i \rightarrow} \ \forall (i,j)  \in ([m] \times [n]) \setminus \{(m,n)\}
\end{align*} 
To ensure the compatibility properties of the run $c$ we also set
\begin{align*}
west(c_{i,j}) &= east(c_{i,j-1}) = q_{\sigma_i \rightarrow}\ \forall \ 2 \leq j \leq n, i \in [m] \\
north(c_{i,j}) &= south(c_{i-1,j}) = q_{\sigma_{i-1} \rightarrow} \ \forall\  2 \leq i \leq m, j \in [n]
\end{align*} 

Now, if $p$ is not fully white or fully black, then $c$ is the only possible accepting run on this image, since every non-border rule $c_{i,j}$ can only output $east(c_{i,j}) = south(c_{i,j})= q_{p_{i,j}\rightarrow}$.  Thus $f_A(p)=1$. The proof is analogous if $p$ is vertically striped. Lastly, in the case where the image is fully white (resp. fully black), we have $f_A(p) = 2$, since an accepting run can either assign $\wver$ or $\whor$ (resp. $\bver$ or $\bhor$) to every pole of the non-border labels.

Now we show that if the automaton $A$ accepts $p$, then $p \in BS$. To show the contrapositive, assume $p \notin BS$. Then, somewhere in $p$ there must exist a $2 \times 2$ square of labels of the form 
$\star\ a \atop a\ b$
where $\star \in \{a,b\}$~(up to rotation and/or bit flip of the symbols). 
Since (i) there are no rules allowing both the north and west poles of $b$ to simultaneously have poles $\wver$ or $\whor$ and (ii) the 
south and east poles of the $a$'s are necessarily in one of these states, there does not exist an accepting run for $p$ and $f_A(p)=0$.
This concludes the proof. 
\end{proof}

It follows from Propositions~\ref{prop:wpa} and~\ref{prop:bs} that there exists a GWM recognizing the BS language. In Section~\ref{learning} we  show how to learn a GWM computing this target function via empirical risk minimization. However, we observe that the particular GWM representation associated to the WPA given in the previous proposition is not unique, by a reasoning similar to the proof of~\citep[Proposition 10]{rabusseau2016thesis}: loosely speaking, the value corresponding to a tensor network associated with a given picture is invariant under  a change of basis of every tensor. More specifically, for any invertible matrix $\P$, we can consider the transformation of the tensors $\ten{T}^\sigma$ described in Figure~\ref{fig:changeofbasis}, with analogous transformations for the border vectors $\ten{\alpha}^{\tau}$ for $ \tau \in \{n,e,s,w\}$. 
One can verify that applying this transformation to all the tensors in a GWM gives a new model which computes the same value on all pictures since all $\mat{P}$ matrices will contract with their inverses.
\begin{figure}[h]
	\centering
	\begin{tikzpicture}
	\input{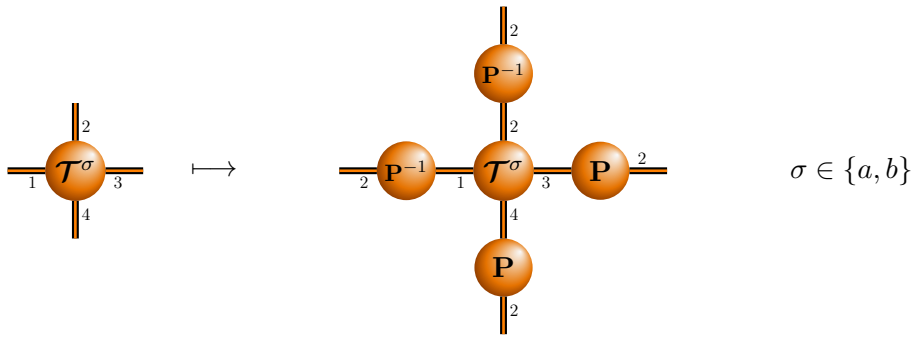}
	\scalebox{1}{  
		
		\node[tensor] (v1)  {$\ten{T}^{\sigma}$};
		\node[draw=none, above = 0.5cm of v1](T-2) {};
		\node[draw=none, below = 0.5cm of v1](T-4) {};
		\node[draw=none, left = 0.5cm of v1](T-1) {};
		\node[draw=none, right = 0.5cm of v1](T-3) {};
		\node[draw=none, right = 1cm of v1](mapsto) {$\longmapsto$};
		\node[tensor, right = 3cm of mapsto] (v2)  {$\ten{T}^{\sigma}$};
		\node[tensor, below = 0.5cm of v2] (v21)  {$\mat{P}$};
		\node[draw=none, below = 0.5cm of v21](P-1) {};
		\node[tensor, right = 0.5cm of v2] (v22)  {$\mat{P}$};
		\node[draw=none, right = 0.5cm of v22](Q-1) {};
		\node[tensor, left = 0.5cm of v2,scale=0.82] (v23)  {$\mat{P}^{-1}$};
		\node[draw=none, left = 0.5cm of v23](Qinv-1) {};
		\node[tensor, above = 0.5cm of v2,scale=0.82] (v24)  {$\mat{P}^{-1}$};
		\node[draw=none, above = 0.5cm of v24](Pinv-1) {};
		
		\node[draw=none, right=2cm of v22] (sigma) {$\sigma \in \{a,b\}$};

		\simpleedgeports{v1}{1}{below left}{T-1}{}{};
		\simpleedgeports{v1}{2}{above right}{T-2}{}{};
		\simpleedgeports{v1}{3}{below right}{T-3}{}{};
		\simpleedgeports{v1}{4}{below right}{T-4}{}{};
		\simpleedgeports{v21}{2}{below right}{P-1}{}{};
		\simpleedgeports{v24}{2}{above right}{Pinv-1}{}{};
		\simpleedgeports{v22}{2}{above right}{Q-1}{}{};
		\simpleedgeports{v23}{2}{below left}{Qinv-1}{}{};
		\simpleedgeports{v2}{4}{below right}{v21}{}{};
		\simpleedgeports{v2}{3}{below right}{v22}{}{};
		\simpleedgeports{v2}{1}{below left}{v23}{}{};
		\simpleedgeports{v2}{2}{above right}{v24}{}{};
		
%
%
%
%
%
%
	}
	\end{tikzpicture}	
	\caption{Transforming all tensors $\ten{T}^\sigma$ through this map, with an analogous map for the border tensors $\ten{\alpha}^{\tau}$, gives a GWM which computes the same value for every picture.}\label{fig:changeofbasis}
\end{figure}


We note that for fixed $m,n$ the Shifting Bits language $SB_{m \times n}$ is also recognizable -- intuitively speaking we can use $n$ states to count that the shift property is being preserved in each row with a construction similar to the previous proposition. By ~\citep[Proposition 5]{rabusseau2016thesis}, finite unions $\cup_{i=1}^k S_{m_i \times n_i}$ are also recognizable. However, we will not be leveraging a recognizability result for this language, since rather than trying to recover an automaton that recognizes the language, we will use this language as an example to show how GWMs can be used to model conditional probabilities via methods described in Section~\ref{learning}.

\subsection{The Learning Problem}\label{learning}

We consider the problem of approximating a target function from a finite number $N$ of input/output examples $S = \{\large(p_i,f(p_i)\large)\}_{i =1}^N \subset \mathcal{P}\times\mathbb{R}$, where $f$ is the target function of interest. To find a candidate GWM $\hat{M}$ which approximates the desired function, we optimize an error function over the training samples with gradient descent. We now outline the specifics for the regression and classification tasks, respectively. 
\\
For the Bars \& Stripes language we seek $f = f_M$ where $f_M$ is the function computed by the GWM recognizing BS from Proposition~\ref{prop:bs}, i.e. the function which outputs 0 on pictures not in BS, 1 on pictures which have only horizontal or vertical stripes, and 2 on pictures which are fully white or fully black. In this setup the candidate GWM $\hat{M}$ is such that $f_{\hat{M}} \simeq f_M$. We minimize the mean squared error (MSE) over the training data.
\begin{equation*}\label{eq:MSE}
J_{\text{MSE}}(\hat{M},S) = \frac{1}{|S|} \sum_{(p,f_M(p))\in S} (f_M(p)-f_{\hat{M}}(p))^2
\end{equation*}
\\
For the Shifting Bits language, we consider the binary classification task which is to predict $1$ on pictures in $SB$ and $0$ otherwise.   As with other commonly used classification models, such as logistic regression and neural networks~\citep{bishop:2006:PRML}, we use a sigmoid function $\sigma:x\mapsto 1 / (1 + e^{-x})$ at the output of the GWM in order to model the conditional probability that the picture $p$ is in SB: $\mathbb{P}\left(f(p_i)=1\vert p_i\right) \simeq \sigma(f_{\hat{M}}(p_i))$. The prediction made in the classification setting is the class with the highest probability, obtained by rounding the output of the sigmoid. The target function is thus $f = \mathbbm{1}_{[SB]}$, the indicator function for the Shifting Bits set. 
 To learn the conditional probability, we use maximum likelihood estimation, which is equivalent to minimizing the cross entropy~(CE) between the target function and the candidate 
\begin{equation}\label{eq:cross-ent}
J_{\text{CE}}(\hat{M},S) = -\frac{1}{|S|} \sum_{(p,f(p))\in S} f(p)\log\sigma\left(f_{\hat{M}}(p)\right) + \left(1-f(p)\right)\log\left(1-\sigma\left(f_{\hat{M}}(p)\right)\right).
\end{equation}

The optimization of both error functions is done using a stochastic gradient descent approach. More precisely, at each iteration $t$, a random mini-batch $S_t$ is uniformly drawn from
$S$ and each parameter of $\hat{M}=(d,\mathcal{T}^a,\mathcal{T}^b,\alpha^w,\alpha^e,\alpha^s,\alpha^n)$ is updated by taking a small step in the opposite direction of the gradient. We use the Adam optimizer~\citep{ADAM} to update the parameters. It is worth mentioning that while deriving an analytic expression for the gradient $\nabla_{(\mathcal{T}^a,\mathcal{T}^b,\alpha^w,\alpha^e,\alpha^s,\alpha^n)} J(\hat{M},S_t)$ may 
be a tedious task, modern deep learning frameworks such as PyTorch~\citep{PYTORCH} can make the implementation of the minimization algorithm relatively uncomplicated by using automatic differentiation techniques to numerically evaluate the gradient via backpropagation.
We remark that the optimization problem is highly non-convex with respect to the weights of the model due to the numerous multiplicative interactions involved in the computations made by the GWM (recall equation \ref{eq:gwm}). Moreover, by the non-uniqueness remarks of Section~\ref{sec:bars}, we can perform a change of basis in the style of Figure~\ref{fig:changeofbasis} to obtain an equivalent model with the same error. This implies that there exist infinitely many global minima to the objective functions, and that if there exists at least one local minimum, then there also exist infinitely many. Thus, convergence to a global minimum is not guaranteed when using gradient descent methods to solve these optimization problems. 

\section{Experiments}
\label{sec:exp}
We implemented the learning procedures detailed in Section~\ref{learning}, and performed experiments on training samples generated from the Bars \& Stripes and the Shifting Bits languages. Our experiments for both languages consider the effect of varying training set sizes and the ability of the models to generalize to unseen pictures of different image dimensions. Our findings are the following: given a sufficient number of training pictures, the learned GWMs are able to (i) make accurate predictions on unseen pictures of the same size and (ii) generalize and make accurate predictions on higher image sizes. 
First, we remark that the datasets are highly class imbalanced: for pictures of size $m \times n$ there are only $2^m + 2^n -2$ and $n2^n -1$ samples (out of $2^{mn}$) which are positive in the Bars \& Stripes and Shifting Bits languages, respectively. To ensure that the model does not learn the constant function $f=0$, we include a 50\% split of positive and negative pictures in the training and test sets. The training and testing sets are kept disjoint whenever possible~(i.e. whenever there are enough positive samples to allow for no overlap).
In all experiments we use the Adam optimizer. We begin by detailing the setup and the experimental results for the Bars \& Stripes language.

\subsection{Bars \& Stripes}

The training and test sets are drawn from the same distribution, however we ensure that the negative samples of the two sets are disjoint so as to accurately test the generalization capabilities of the model. Note that it is not possible to have disjoint positive samples due to their sparsity (e.g. $4 \times 4$ pictures only have $30$ positive images), thus we expect the models to correctly predict on all possible positive samples since they likely will have seen all of them. The real challenge of the task is to learn a model which correctly classifies all other images as negative.
In addition to the mean squared error,  we also report the classification accuracy of our model obtained via thresholding: the model classifies the image as positive if the predicted output is greater than 0.5 and classifies the image as negative otherwise.
In all experiments, we  initialize the tensor values with values independently drawn from a Gaussian distribution with mean zero and standard deviation $0.4$. All models are trained with tensors of dimension $d=6$ to allow the learning procedure to recover the automaton of Proposition \ref{prop:bs}.

Our first experiment considers the effect of varying training set sizes. Here, all pictures are of size $4 \times 4$, the size of the dataset is varied from $N=500$ to $N=50{,}000$, and the test set is of fixed size $N_{test}=100$. 
We use a learning rate of $0.01$ and mini-batches of size $100$. The training errors, testing errors, and classification accuracies for different values of $N$ are reported in~Table \ref{table:exp1} and a plot of the mean squared error~(MSE) at each iteration for the case $N = 10{,}000$ is shown in Figure~\ref{fig:train1}. We observe that even for small training set sizes, the model is able to reduce the MSE and attain perfect (or near-perfect) classification accuracy after seeing a relatively small number of mini-batches, although the training and testing errors do increase as $N$ decreases.
\begin{figure}[h]
\begin{floatrow}
\ffigbox{%
 \includegraphics[scale=0.6, clip=true]{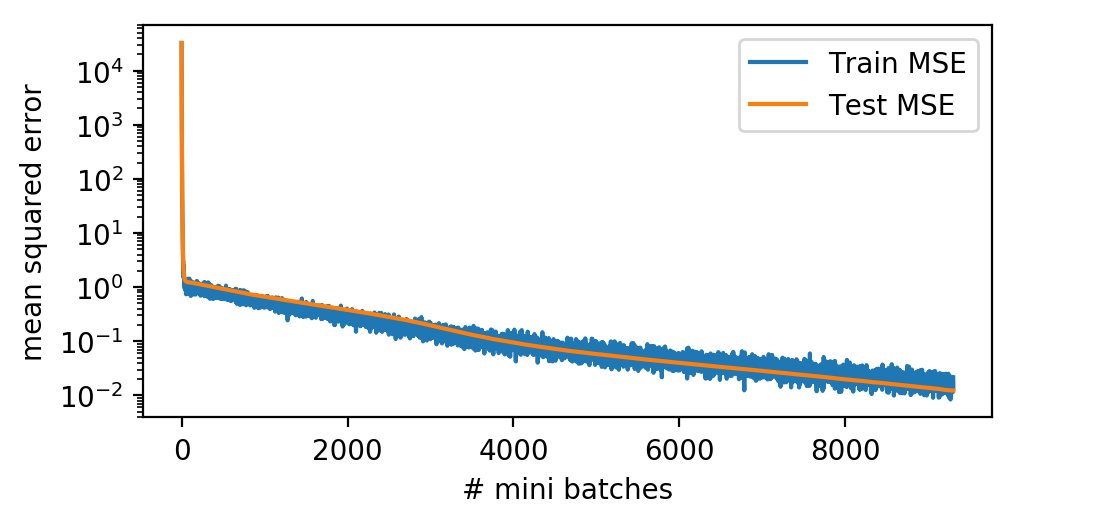}%
}{%
\vspace*{-0.7cm}
 \caption{MSE vs. number of mini-batches, when trained on $10{,} 000$ pictures of size $4 \times 4$. The final MSEs are reported in Table \ref{table:exp1}. Note that the MSE axis is log scaled.}%
 \label{fig:train1}
}
\capbtabbox{%
	\scalebox{0.82}{
 \begin{tabular}{|c | c c c|} 
 \hline
 N & Training error & Testing error & Accuracy \\ [0.5ex] 
 \hline\hline
 $50{,}000$ & $4.5\cdot10^{-4}$ & $7.6\cdot10^{-4}$ & $100\%$ \\ 
 \hline
 $25{,}000$ & $7.5\cdot10^{-3}$ & $9.4\cdot10^{-3}$ & $100\%$ \\
 \hline
 $10{,}000$ & $1.1\cdot10^{-2}$ & $1.2\cdot10^{-2}$ & $100\%$ \\
 \hline
 $5{,}000$ & $6.0\cdot10^{-3}$ & $7.9\cdot10^{-3}$ & $100\%$ \\
 \hline
 $1{,}000$ & $2.96\cdot10^{-3}$ & $1.4\cdot 10^{-2}$ & $100\%$ \\ 
 \hline
 $500$ & $2.4\cdot10^{-3}$ & $2.2\cdot10^{-2}$ & $99\%$ \\[0.5ex] 
 \hline
\end{tabular}}
}{%
  \caption{Training error, testing error, and classification accuracy for varying sizes of the training set.}%
  \label{table:exp1}
}
\end{floatrow}
\end{figure}

Despite the low test and misclassification errors on $4 \times 4$ pictures, the learned models are unable to predict accurate values for pictures of larger size. For instance, when applied on a test set of $5 \times 5$ pictures, the GWM learned on the training set of size $50{,}000$ yielded a mean squared error of $203.2$ and an accuracy of $57\%$ (recall that the data is a $50\%$ split, so this accuracy is just slightly better than random). This is evidence that we have not recovered the GWM of proposition \ref{prop:bs} (or an equivalent model computing the same function), since that automaton recognizes pictures of all sizes and thus would indeed have generalized.

\begin{figure}[h!]\label{fig:xp2}
\begin{center}
\includegraphics[width=0.46\textwidth]{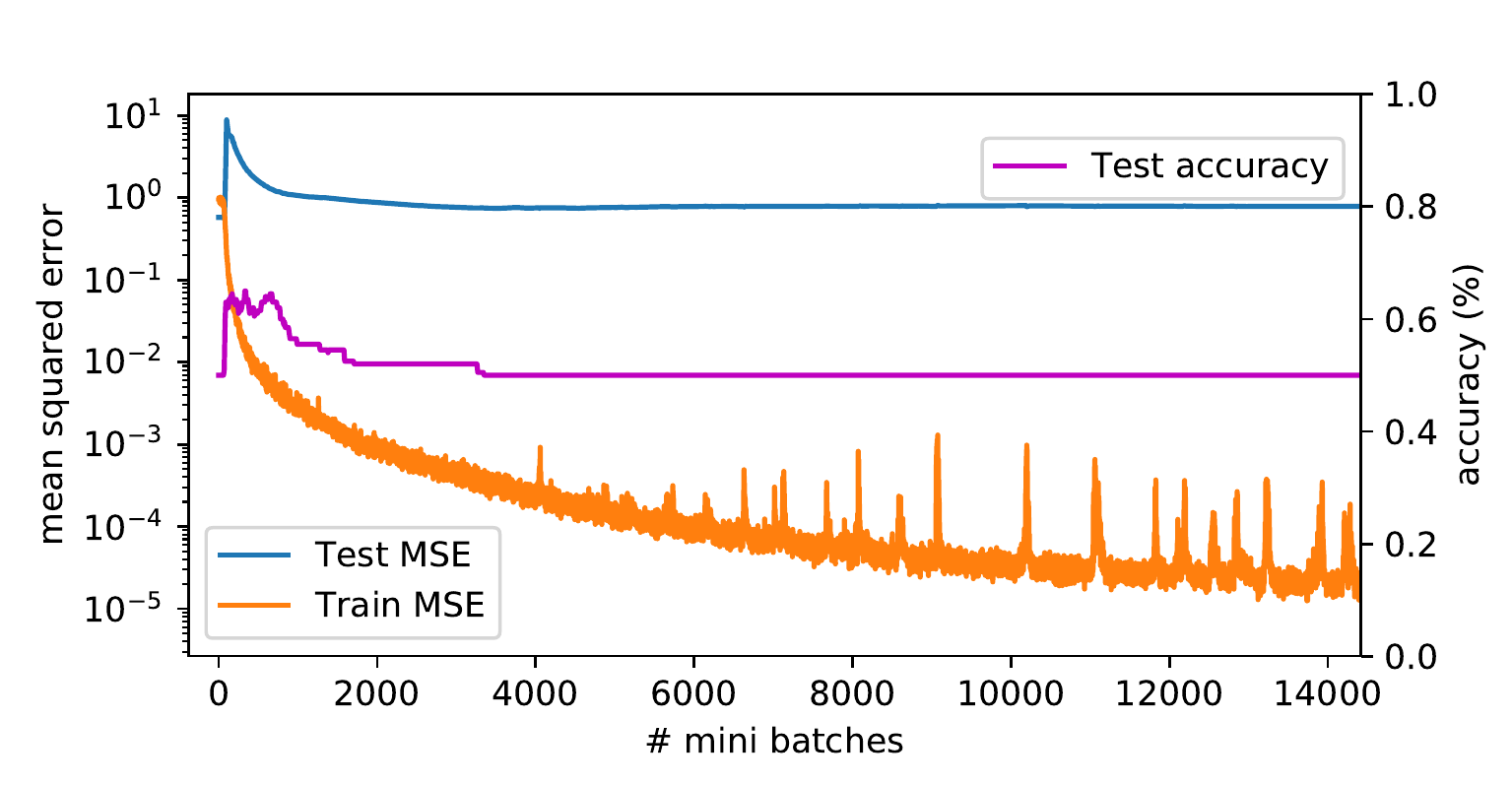}\hspace*{0.5cm}\includegraphics[width=0.46\textwidth]{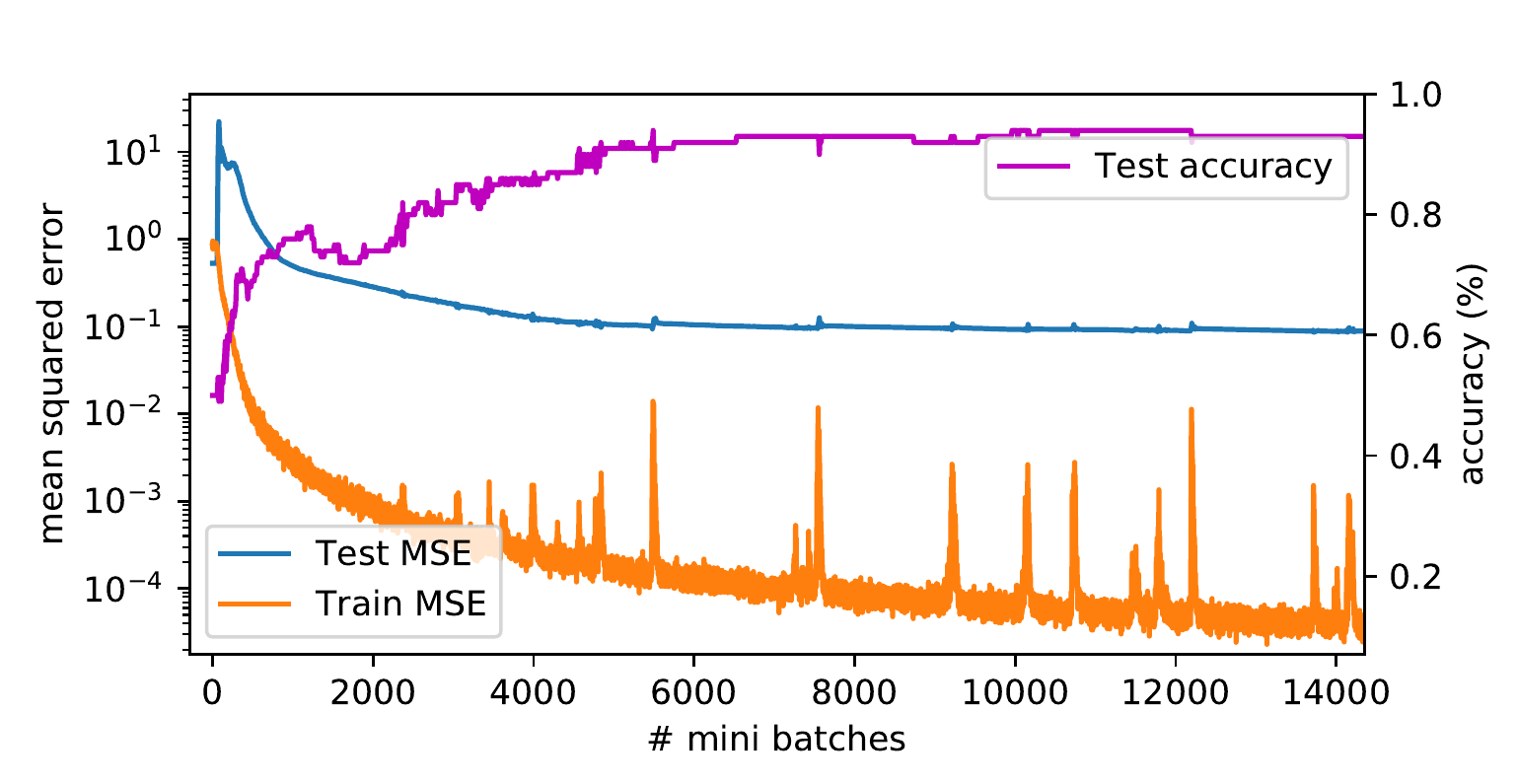}
\end{center}
\caption{Training error, testing error and classification accuracy when training on pictures of sizes $2\times2$ up to $4\times 4$ and testing on $5\times 5$ pictures (left) and when training on pictures of sizes $2\times2$ up to $5\times 5$ and testing on $6\times 6$ pictures (right).}\label{fig:xp2}
\end{figure}

The next experiment remedies this problem by combining multiple picture sizes together in the training set. Specifically, the training set is obtained by sampling $10{,}000$ pictures~($5{,}000$ positives and $5{,}000$ negatives) of different sizes ranging from $2 \times 2, 3\times 3, \cdots, m \times m$. The test set contains 200 pictures of size $(m+1)\times (m+1)$. We run this experiment for $m=4$ and $m=5$ and use a learning rate of 0.001 and mini-batches of size 1,000.
The train and test MSE and the test accuracies are plotted as a function of the number of iterations in Figure~\ref{fig:xp2} where we see that the learned model is able to generalize for the case of $m=5$ but not for $m=4$. The jump in performance observed from the case of $m=4$ to the case of $m=5$ is attributed to the greater wealth of data available in the training set. This experiment suggests that, given enough data, minimizing the mean squared error over the training data using gradient-based methods allows one to recover a GWM which is a good approximation of the target picture series. 
\subsection{Shifting Bits}

Next, we investigate the performance of our methods on classification tasks with samples generated from the Shifting Bits language.
In this language, despite the class imbalance, there are enough positive samples ($n2^n-1$ out of $2^{mn}$) so that the training and testing sets can be fully disjoint for both the positive and the negative samples. We initialize the tensor values with samples drawn from a Gaussian distribution with mean zero and standard deviation $0.2$, and all models have dimension $d=10$. Unlike the previous experiments, here the dimension is treated as a hyperparameter since we are not making use of a recognizability result. 
The learning rate was set to 0.01, with mini-batches of size 128, and we used gradient clipping due to the greater instability observed when minimizing the cross-entropy with the use of the sigmoid function. 

We trained models on multiple datasets of size $N$ consisting of pictures with fixed height $m$ and different widths $n$, and tested the ability of the model to generalize to pictures of unseen sizes. Specifically, we experimented with datasets of sizes ranging from $N=1000$ to $N=20{,}000$, with one model trained on pictures of sizes $2 \times 5, 2 \times 6, \dots, 2 \times 15$ and the other trained on pictures of sizes $3 \times 5, 3 \times 6, \dots, 3 \times 15$. For each model, we tested on $N_{test}=200$ pictures of widths $10, 20, 50,$ and $100$. The width of each picture in the training set was uniformly drawn from $5,\dots,15$. The testing accuracy for different pictures as a function of the training height and dataset size is given in Table \ref{table:shift-xp1} and a plot of the different testing accuracies at each epoch of training for the case of $m=2, N=5{,}000$ is shown in Figure \ref{fig:classif1}. Note that an epoch corresponds to a full pass through the training set (i.e. $N/128$ number of mini-batches). The figure shows that the model quickly converges to high accuracies on all $2 \times n$ sets, and correspondingly minimizes the cross-entropy errors (not shown). 
\begin{figure}[h]
\begin{floatrow}
\ffigbox{%
\hspace*{-0.5cm}
 \includegraphics[width=0.55\textwidth, clip=true]{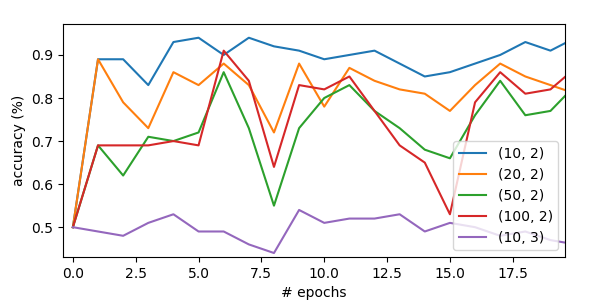}%
}{%
\vspace*{-0.65cm}
 \caption{Accuracy vs. epochs when tested on different picture sizes and trained on $5{,} 000$ pictures of height $2$ and widths uniformly drawn between $5,\dots,15$. The final accuracies are reported in Table \ref{table:shift-xp1}.}%
 \label{fig:classif1}
}

\capbtabbox{%
	\scalebox{0.79}{
	\hspace*{-0.5cm}
 \begin{tabular}{|c | c c c c c |}
 \hline
$m=2, N=$ & $2 \times 10$ & $2 \times  20$ & $2 \times  50$ & $2 \times  100$ & $3 \times  10$ \\ 
 \hline\hline 
1000& $67.0\%$ & $67.0\%$ & $61.0\%$ & $59.0\%$ & $61.0\%$ \\ 
5000& $94.0\%$ & $90.0\%$ & $73.0\%$ & $69.0\%$ & $58.0\%$ \\ 
20000& $100.0\%$ & $96.0\%$ & $96.0\%$ & $92.0\%$ & $53.0\%$\\ 
 \hline 
 \hline
$m=3, N=$ & $3 \times 10$ & $3 \times  20$ & $3 \times  50$ & $3 \times  100$ & $2 \times  10$ \\ 
 \hline\hline 
1000& $49.0\%$ & $48.0\%$ & $60.0\%$ & $50.0\%$ & $58.0\%$ \\ 
5000& $45.0\%$ & $51.0\%$ & $50.0\%$ & $51.0\%$ & $48.0\%$ \\ 
20000& $100.0\%$ & $100.0\%$ & $97.0\%$ & $100.0\%$ & $50.0\%$\\ 
 \hline 
 \end{tabular}
}
}{%
  \caption{Testing accuracies on different picture sizes for varying dataset sizes, for models trained on pictures of height $2$ (above) and height $3$ (below).}%
  \label{table:shift-xp1}
}

\end{floatrow}
\end{figure}

Given enough data, the learned models for $m=2$ and $m=3$ are able to generalize to wider pictures of the same height, attaining accuracies greater than $90\%$, although for the case $m=3$ more training data is needed to obtain accurate predictions. In addition to the training and test sets being disjoint, we note that the models are tested on pictures that can have shifts of size greater than 15~(while all training pictures had, at most, a shift of 15). However, we observe that the models trained on size $m=2$ (resp. $m=3$) are unable to correctly classify pictures of height $3$ (resp. $2$), and can only attain accuracies of around $50\%$. 

In an effort to train a model which would generalize to multiple heights, we attempted to combine pictures of heights 2 and 3 (still with multiple widths) together in the training set, and to test on pictures of height $ 2$, $3$, and $4$. However, unlike in the Bars \& Stripes experiments, we found that it was possible to minimize the training error but to not correctly classify all the test sets. In our experiments the models would only settle for classifying a single picture size to near-perfect accuracies and attain $50\%$ on other sizes. Some moderate but non-exhaustive exploration was done on the hyperparameters (i.e. learning rate and dimension of the model) with no significant impact on performance. We leave further experimentation for future work.

\section{Conclusion}
\label{sec:conc}

We investigated the use of gradient-based algorithms for learning  GWMs from data and provided empirical evidence suggesting that it is possible to recover a GWM which approximates a function recognizing formal picture languages. We tackled regression and classification tasks and showed that in both cases the training and testing errors can be minimized leading to models which generalize to unseen pictures of different sizes. Our experiments on the toy Bars \& Stripes and Shifting Bits picture languages indicate that this is a promising direction, although this is just a first step. In the future, we intend to apply these methods to more challenging tasks, including more traditional machine learning problems which have typically not been viewed from the formal languages perspective. Furthermore, these methods can equally be applied to other families of graphs which could lead to interesting applications (e.g. in bioinformatics or in natural language processing). Moreover, a theoretical analysis of this approach remains to be done and is of particular interest. Extending the spectral learning algorithms to WPA and GWMs is also a direction that we will continue to investigate.

\newpage
\bibliography{gwm-images}

\end{document}